\documentclass{IEEEtran4PSCC}
\ifCLASSINFOpdf
   \usepackage[pdftex]{graphicx}
\else
   \usepackage[dvips]{graphicx}
\fi
\usepackage[cmex10]{amsmath}
\usepackage{amssymb}
\usepackage{verbatim}
\usepackage{xurl} % allows line breaks anywhere in URLs
\usepackage{caption}  % for captions
\usepackage{subcaption}
\usepackage{url}
\makeatletter
\let\old@ps@headings\ps@headings
\let\old@ps@IEEEtitlepagestyle\ps@IEEEtitlepagestyle
\def\psccfooter#1{%
    \def\ps@headings{%
        \old@ps@headings%
        \def\@oddfoot{\strut\hfill#1\hfill\strut}%
        \def\@evenfoot{\strut\hfill#1\hfill\strut}%
    }%
    \def\ps@IEEEtitlepagestyle{%
        \old@ps@IEEEtitlepagestyle%
        \def\@oddfoot{\strut\hfill#1\hfill\strut}%
        \def\@evenfoot{\strut\hfill#1\hfill\strut}%
    }%
    \ps@headings%
}
\makeatother

\psccfooter{%
        \parbox{\textwidth}{\hrulefill \\ \small{24th Power Systems Computation Conference} \hfill \begin{minipage}{0.2\textwidth}\centering \vspace*{4pt} \includegraphics[scale=0.06]{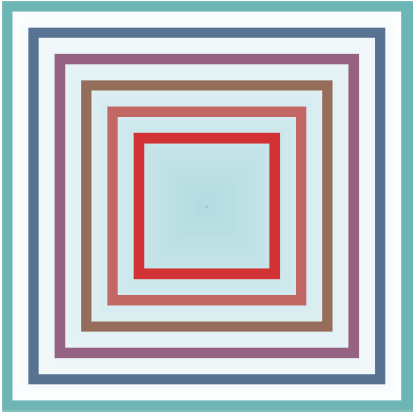}\\\small{PSCC 2026} \end{minipage} \hfill \small{Limassol, Cyprus --- June 8-12, 2026}}%
}
\usepackage{amsthm}
\newtheorem{thm}{Theorem}[section]
\newtheorem{prop}[thm]{Proposition}
\newtheorem{rem}{Remark}

\usepackage{xcolor}

\begin{document}
\title{Decentralized Small Gain and Phase Stability Conditions for Grid-Forming Converters: Limitations and Extensions}

%% To specify the authors when (number of affiliations <= 2)
\author{
\IEEEauthorblockN{Diego Cifelli, Adolfo Anta}
\IEEEauthorblockA{AIT Austrian Institute of Technology,
Vienna, Austria
%\\
%\{diego.cifelli, adolfo.anta\}@ait.ac.at}
}}

%% To specify the authors when (number of affiliations > 2)
% \author{\IEEEauthorblockN{Author n.1\IEEEauthorrefmark{1},
% Author n.2\IEEEauthorrefmark{2},
% Author n.3\IEEEauthorrefmark{3}, 
% Author n.4\IEEEauthorrefmark{3} and
% Author n.5\IEEEauthorrefmark{4}}
% \IEEEauthorblockA{\IEEEauthorrefmark{1} Department Name of Organization A\\
% Name of the organization A,
% Address A\\ Emails if wanted}
% \IEEEauthorblockA{\IEEEauthorrefmark{2} Department Name of Organization B\\
% Name of the organization B,
% Address B\\ Emails if wanted}
% \IEEEauthorblockA{\IEEEauthorrefmark{3} Department Name of Organization C\\
% Name of the organization C,
% Address C\\ Emails if wanted}
% \IEEEauthorblockA{\IEEEauthorrefmark{4}Department Name of Organization D\\
% Name of the organization D,
% Address D\\ Emails if wanted}
% }

\maketitle

\begin{abstract}
The increasing share of converter-based resources in power systems calls for scalable methods to analyse stability without relying on exhaustive system-wide simulations. Decentralized small-gain and small-phase criteria have recently been proposed for this purpose, but their applicability to grid-forming converters is severely limited by the sectoriality assumption, which is not typically satisfied at low frequencies. This work revisits and extends mixed gain–phase conditions by introducing loop-shaping transformations that reformulate converter and network models in alternative coordinate frames. The proposed approach resolves intrinsic non-sectoriality at low frequencies and reduces conservativeness, thereby improving the applicability of decentralized stability certification. Analytical results are illustrated using an infinite bus system first and then extended to the IEEE 14-bus network, demonstrating the practicality and scalability of the method. These findings provide a pathway toward less conservative and more widely applicable decentralized stability certificates in power grids.
\end{abstract}

\begin{IEEEkeywords}
converter-dominated grids; small-phase criterion; decentralized stability conditions; grid-forming converters.
\end{IEEEkeywords}
% Use this to place sponsorships
\thanksto{This work has been
funded by the CETPartnership, the Clean Energy Transition Partnership under the 2023 joint call for research proposals, co-funded by the European Commission (GA N 101069750) and with the funding organizations AEI (National State Research Agency, Spain), PtJ (MWIKE) (Projekträger Jülich/Forschungszentrum Jülich GmbH, Germany), GSRI (General Secretariat for Research and Innovation, Greece), SEAI (Sustainable Energy Authority of Ireland), FFG (Austrian Research Promotion Agency) and RCN (The Research Council of Norway).}

\section{Introduction}
As the penetration of power electronic converters in modern power systems continues to rise, fundamental changes in grid dynamics are becoming increasingly evident. Interactions between converter-based resources and the existing network introduce complex behaviors that remain poorly understood. In practice, such converter-driven dynamics have already led to unexpected and undesirable phenomena~\cite{Cigre2024}, underscoring the need for systematic stability assessment. Traditionally, ensuring stability has relied on extensive simulation studies, which become computationally prohibitive and impractical when scaling to systems with hundreds of converters. Moreover, each time a new device is connected, the entire assessment must be repeated, further amplifying the challenge.

To overcome these limitations, recent research has proposed decentralized stability criteria~\cite{chenExtendedFrequencyDomainPassivity2025a,huangGainPhaseDecentralized2024,häberle2025decentralizedparametricstabilitycertificates} that enable the impact of a single converter to be assessed independently of others in the system. These conditions are designed such that, if each individual converter satisfies them, the overall system stability is guaranteed without requiring explicit analysis of converter-converter interactions. This paradigm shift offers a scalable and modular pathway to stability certification in converter-dominated grids, reducing reliance on exhaustive system-wide simulations. Moreover, these decentralized conditions yield valuable insights for converter control design to ensure system stability. 

%- Generic Introduction (Adolfo)
%- Problem presentation (Adolfo)
%- State of Art - Limbin, Verena, Indians, KTH, Eder, ...\\

A well-known decentralized criterion for ensuring stability requires that all system components exhibit passivity. This requirement is starting to appear in upcoming grid codes for grid-forming converters \cite{fingridGridCode}. Although this condition is naturally satisfied by the network, it does not hold for converters. At high frequencies, passivity can be enforced through suitably control strategies. 
At low frequencies, understood here in the synchronous reference (dq) frame, non-passivity is intrinsically associated with the constant-power characteristics of converters, a property that cannot be removed. 

While passivity provides a conservative stability criterion, subsequent research has introduced more general and less restrictive frameworks.  In \cite{chenExtendedFrequencyDomainPassivity2025a}, and extended in~\cite{chenUnifiedFlexibleFrequencyDomain2025a}, passivity indices allow frequency-wise compensation between passive and non-passive devices. Extended passivity indices with weighting matrices offer additional flexibility and reduced conservatism. However, these approaches focus on single converters and have not been extended to multi-converter systems.

Alternative decentralized conditions were proposed in~\cite{huangGainPhaseDecentralized2024} using more sophisticated methods, based on the novel definition of phase for a matrix to elaborate a decentralized mixed gain and phase criteria. However, the converter admittance needs to satisfy a sectoriality condition that is not fulfilled in practice
%~\cite{chenExtendedFrequencyDomainPassivity2025a}, 
limiting its applicability. A similar approach appeared in \cite{woolcockMixedGainPhase2023} aimed at robust analysis.

To overcome the sectoriality assumption, \cite{baron2025decentralized} proposed an alternative method for phase computation based on scaled relative graphs (SRGs). While promising, the construction of such SRGs is computationally demanding. Moreover, computing SRGs with black-box models is still an open point.

Another relevant study \cite{deyPassivityBasedDecentralizedCriteria2023} formulates the system in a polar reference frame, in which grid-forming (GFM) converters are passive at low frequencies. Assuming time-scale separation, the authors propose a hybrid criterion: passivity is verified in the polar frame at low frequencies, while the standard rectangular admittance representation is used at high frequencies. However, the time-scale separation assumption may not hold in practice or may be difficult to verify. Decentralized parametric conditions on P-f and Q-V control were also proposed in \cite{häberle2025decentralizedparametricstabilitycertificates}. However, the scope is limited to predefined control structure and ignore lower level converter dynamics.

In this paper, focusing on the decentralized small gain/phase condition proposed in \cite{huangGainPhaseDecentralized2024}, we first identify some limitations of the method when applied to grid-forming converters, which exhibit high gain and lack sectoriality at low frequencies. We then analyze the origin of this non-sectoriality, showing that it arises from the constant-power characteristic of converters and is therefore unavoidable. Finally, inspired by \cite{deyPassivityBasedDecentralizedCriteria2023}, we extend the framework of \cite{huangGainPhaseDecentralized2024} by introducing a loop-shaping transformation, which generalizes the approach, provides less conservative conditions, and resolves the sectoriality issue.
%In this work, we extend the work in \cite{huangGainPhaseDecentralized2024}, taking inspiration from \cite{deyPassivityBasedDecentralizedCriteria2023}, by generalizing the approach introducing loop-shaping transformation allowing a more flexible conditions and overcoming the sectorial issue.   
%If there is space, paper structure.

%\color{red}
The remainder of this paper is organized as follows: Section \ref{sec:gain_phase_crit} reviews decentralized gain-phase stability criteria, while Section \ref{sec:lim} addresses their limitations regarding GFM converters. Section \ref{sec:propo_approch} details our proposed extensions. The methodology is validated in Sections \ref{sec:ex_1} and \ref{sec:ex_2} for a converter connected to an infinite bus and the IEEE 14 Bus system, respectively. Finally, Section \ref{sec:conc} concludes the paper.
%\color{black}

Data and code to reproduce the results are available at \cite{RepoZenodo}.

\section{Mixed Gain-Phase Stability Conditions}
\label{sec:gain_phase_crit}

For a matrix $A\in \mathbb{C}^{n\times n}$, the $n$ gains are defined as its singular values, with $\overline{\sigma}(A)$ and $\underline{\sigma}(A)$ being the largest and smallest gains, respectively. These quantities can be used to bound both the eigenvalues of $A$ and the magnitudes of the eigenvalues of the product of two matrices $A,B\in\mathbb{C}^{n\times n}$, as in~\eqref{eq:sv_bound}. Importantly, these bounds are independent of the explicit computation of the product $AB$.
%\begin{equation}
%    \overline{\sigma}(A) = \sigma_1(A) \geq \cdots \geq \sigma_n(A) = \underline{\sigma}(A)
%    \label{eq:sv_def}
%\end{equation}
\begin{equation}
    \underline{\sigma}(A)\underline{\sigma}(B) \leq |\lambda(AB)| \leq \overline{\sigma}(A)\overline{\sigma}(B) 
    \label{eq:sv_bound}
\end{equation}
It is desirable to establish analogous bounds for the phases of the eigenvalues. This problem is less straightforward, and only recently has a definition of matrix phase been proposed that enables such bounds \cite{wangPhasesComplexMatrix2020, chenPhaseAnalysisIEEE}.

The numerical range $W(A)$ of a complex matrix $A\in \mathbb{C}^{n\times n}$ is defined by \eqref{eq:num_range}. It is a convex subset of $\mathbb{C}$ that contains the eigenvalues of $A$.
\begin{equation}
    W(A)=\{ x^*Ax:\,x\in\mathbb{C}^n, ||x||=1\}
    \label{eq:num_range}
\end{equation}
A matrix is said to be sectorial if $W(A)$ lies entirely within an angular sector of the complex plane and $0 \notin W(A)$. For a sectorial matrix $A$, there exists an invertible matrix $T$ such that $A=T^*DT$, where $D$ is a unique diagonal matrix (up to permutation). 
%The diagonal elements of $D$ are distributed in an arc on the unit circle with length smaller than $\pi$.
The phases of $A$ are defined as the arguments of the diagonal entries of $D$, ordered so that $\overline{\phi}(A)$ and $\underline{\phi}(A)$ denote the maximum and minimum phases, with $\overline{\phi}(A)-\underline{\phi}(A) < \pi$. The phase interval $[\underline{\phi}(A),\, \overline{\phi}(A)]$ is shortly denoted by $\Phi(A)$. This definition can be extended to quasi-sectorial (i.e. $0 \in \partial W(A)$) and to semi-sectorial (i.e. $0 \in \partial W(A)$ and $\overline{\phi}(A)-\underline{\phi}(A) \leq \pi$) matrices \cite{chenPhaseTheoryMIMO2022}.
%\begin{equation}
%    \overline{\phi}(A)=\phi_1(A)\geq\cdots\geq\phi_n(A)=\underline{\phi}(A)
%    \label{eq:phase_def}
%\end{equation}
Geometrically, these phases correspond to the angular boundaries of the numerical range. In direct analogy with singular values, maximum and minimum phases can be used to bound the phases of the eigenvalues of a matrix product as in~\eqref{eq:phase_bound}.
\begin{equation}
    \underline{\phi}(A)+\underline{\phi}(B) \leq \angle{\lambda(AB)} \leq \overline{\phi}(A)+\overline{\phi}(B) 
    \label{eq:phase_bound}
\end{equation}
\subsection{Mixed Gain-Phase Stability Criterion}
\begin{figure}
    \centering
    \begin{subfigure}{0.45\columnwidth}
        \centering
        \includegraphics[width=\linewidth]{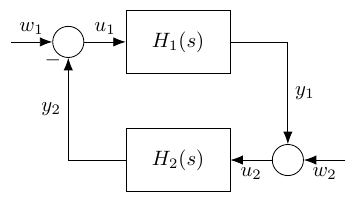}
        \caption{Generic Feedback System}
        \label{fig:fd_diagram}
    \end{subfigure}
    \begin{subfigure}{0.45\columnwidth}
        \centering
        \includegraphics[width=\linewidth]{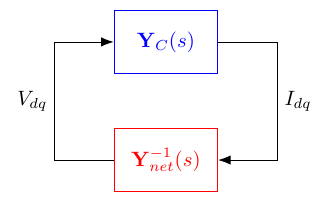}
        \caption{Admittance Feedback}
        \label{fig:fd_diagram_conv_net}
    \end{subfigure}
    \caption{Diagram of the feedback connection.}
    \label{fig:twoside}
\end{figure}
The concepts introduced above provide a framework for assessing the stability of linear time-invariant systems. Consider the feedback interconnection of $H_1(s)$ and $H_2(s)$ shown in Figure \ref{fig:fd_diagram}. 
The following theorem gives a sufficient condition for the stability of the resulting closed-loop system \cite{zhao2022smallgainmeetssmall}.
\begin{thm}[Mixed Small Gain-Phase Theorem]
    Let the open-loop systems $H_1, H_2$ be real, rational, stable, and proper transfer function matrices. Then, the closed loop system is stable if for each $\omega \in [0, \infty]$, either
    \begin{enumerate}
        \item The product of the maximum singular values satisfies:
        \begin{equation*}
            \overline{\sigma}(H_1(j\omega))\overline{\sigma}(H_2(j\omega))<1
        \end{equation*}
        \item 
        or $H_1(j\omega)$ is semi-sectorial and $H_2(j\omega)$ is sectorial, and
        \begin{equation*}
            \overline{\phi}(H_1(j\omega))+\overline{\phi}(H_2(j\omega))<\pi
        \end{equation*}
        \begin{equation*}
            \underline{\phi}(H_1(j\omega))+\underline{\phi}(H_2(j\omega))>-\pi
        \end{equation*}
    \end{enumerate}
    \label{thm:mixed_small_gain_phase}
\end{thm}
This theorem unifies the traditional small-gain condition and its phase counterpart into a single stability framework, enabling the use of either gain-based or phase-based bounds at each frequency, depending on which provides a less conservative estimate.

\subsection{Decentralized condition for power system stability}
\begin{figure}
    \centering
    \includegraphics[width=\linewidth]{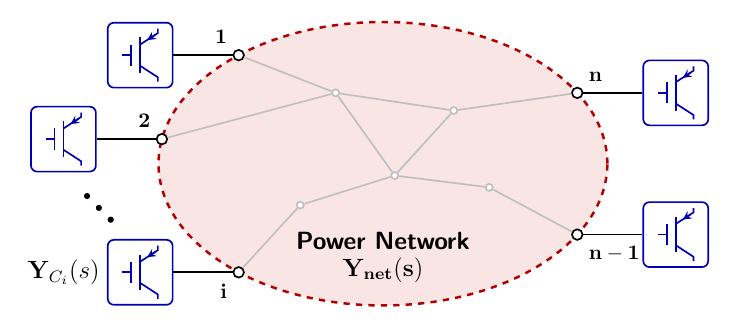}
    \caption{Illustration of a multi-converter power system}
    \label{fig:multi_converter_diagram}
\end{figure}

In the context of power systems, we are interested in assessing the stability of the interconnection of $n$ power electronic converters to a given network, illustrated in Figure \ref{fig:multi_converter_diagram}. 
Let 
\begin{equation}
    \mathbf{Y}_C(s)=\text{diag}(Y_{C_1}(s),Y_{C_2}(s),\cdots,Y_{C_n}(s))
    \label{eq:block_diagonal}
\end{equation}
denote the overall converter admittance matrix, where $Y_{C_i}(s)$ represents the admittance of the $i$-th converter. Let $\mathbf{Y}_{net}(s)$ denote the network admittance. In this work, $\mathbf{Y}_{net}(s)$ is assumed to be the Kron-reduced network, where all internal nodes that are not connected to any converter are eliminated. All admittances are expressed in a common global $dq$ reference frame. Analogous to Figure \ref{fig:fd_diagram},  the interconnection between the converters and the network can be represented as the feedback loop shown in Figure \ref{fig:fd_diagram_conv_net}.

Theorem \ref{thm:mixed_small_gain_phase} can be applied directly to this system to provide a stability condition. However, the block-diagonal structure of $\mathbf{Y}_C(s)$ allows a decentralized stability analysis, leading to the following result \cite{huangGainPhaseDecentralized2024}.

\begin{thm}[Decentralized Mixed Small Gain-Phase Theorem]
    The multi-converter system in Figure \ref{fig:fd_diagram_conv_net} is stable if the open-loop systems are stable, and for each $\omega \in [0,\infty]$, either
    \begin{enumerate}
        \item 
        the decentralized gain condition, i.e.,
        \begin{equation*}
            \max_i \overline{\sigma}(Y_{C_i}(j\omega))<\underline{\sigma}(\mathbf{Y}_{net}(j\omega))
        \end{equation*}
        holds, or
        \item 
        the decentralized phase condition, i.e
        every $Y_{C_i}(j\omega)$ is quasi-sectorial, and $\mathbf{Y}_{net}(j\omega)$ is sectorial, and
        \begin{equation*}
            \overline{\phi}(Y_{C_i}(j\omega))<\pi-\overline{\phi}(\mathbf{Y}_{net}^{-1}(j\omega)) \quad \forall i
        \end{equation*}
        \begin{equation*}
            \underline{\phi}(Y_{C_i}(j\omega))>-\pi-\underline{\phi}(\mathbf{Y}_{net}^{-1}(j\omega))\quad \forall i
        \end{equation*}
        \begin{equation*}
            \max_i\overline{\phi}(Y_{C_i}(j\omega)) -\min_i\underline{\phi}(Y_{C_i}(j\omega))<\pi
        \end{equation*}
    \end{enumerate}
    \label{thm:dec_mixed_small_gain_phase}
\end{thm}
This decentralized formulation allows the stability of a large-scale multi-converter system to be certified by verifying local conditions on each converter model together with the network model, without requiring the explicit computation of the full interconnected system’s dynamics.
\section{Criteria Limitations for GFM converters}
\label{sec:lim}

In this section, we illustrate why the stability conditions of Theorem \ref{thm:dec_mixed_small_gain_phase} are difficult to be satisfied in practical power system applications.
A typical example arises with grid-forming (GFM) converters, which often exhibit high low-frequency gain.
%Under such conditions, the small-gain condition is not satisfied. At the same time, the converter impedance matrix may fail to be sectorial at those same frequencies, rendering the small-phase condition inapplicable.
To demonstrate this limitation, we consider a conventional GFM converter implementing virtual synchronous machine (VSM) power control (described in detail in Appendix \ref{apx:GFM_control}) connected to an infinite bus, see Figure \ref{fig:gfm_diag}, through a network with low short-circuit ratio (SCR), i.e., large grid impedance. The resulting small-gain and small-phase plots for the feedback loop are shown in Figure~\ref{fig:GFM_gain_phase_test}.

Three main issues arise from this example:
\begin{itemize}
    \item Violation of the small-gain condition: the gain of the converter significantly exceeds that of the grid. As a result, the small-gain condition is never satisfied.
    \item Violation of sectoriality: below $0.9$ Hz, the converter impedance is not sectorial, and therefore the small-phase criterion cannot be applied. 
    \item Conservativeness: even when the admittance is sectorial and the phase is well-defined, the phase intervals of the converter and the grid are very close to overlapping for frequencies below $50$ Hz, and eventually do overlap near $1$ Hz. This overlap is not necessarily indicative of an actual instability but is instead a consequence of the conservativeness of the small-phase criterion.
\end{itemize}
%Although the latter issue is less restrictive than the first two, it nonetheless represents a practical limitation that can hinder the applicability of the method. Later in the paper, we revisit this problem and propose a solution to mitigate such conservativeness.
From this example, it is clear that the classical small-gain and small-phase criteria exhibit significant limitations in realistic converter-dominated scenarios. A detailed examination is now presented in the following subsection to identify the source of this non-sectoriality.
\begin{figure}[tbh]
    \centering
    \begin{subfigure}{\columnwidth}
        \centering
        \includegraphics[width=\linewidth]{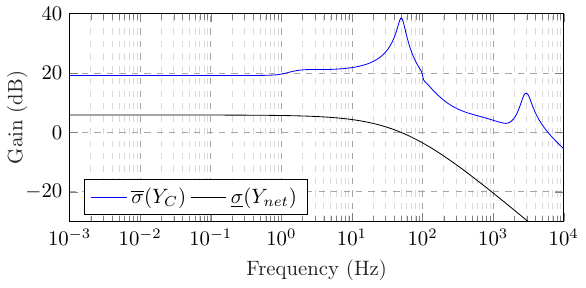}
    \end{subfigure}
    \begin{subfigure}{\columnwidth}
        \centering
        \includegraphics[width=\linewidth]{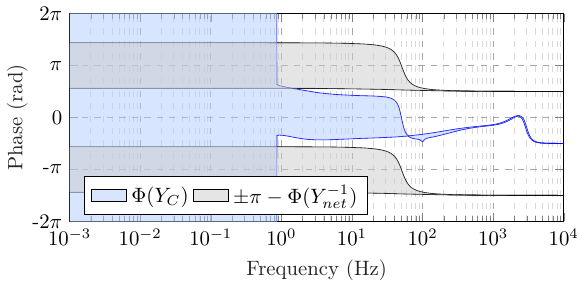}
    \end{subfigure}
    \caption{Mixed small gain–phase condition for a GFM converter. If a system is not sectorial, its phase is represented by the full interval $[-2\pi,\, 2\pi]$. The small-gain condition is violated if the converter gain exceeds the network gain. The small-phase condition is violated if the converter phase overlap the network phase shifted by $\pm\pi$.}
    \label{fig:GFM_gain_phase_test}
\end{figure}

\subsection{Non-Sectoriality of GFM converter}
As illustrated in the previous example, and also reported in other studies, converter admittances are often non-sectorial at low frequencies \cite{huangGainPhaseDecentralized2024,woolcockMixedGainPhase2023}. However, the origin of this phenomenon has not been thoroughly investigated, and it remains unclear whether it is a consequence of specific control architectures and parameter choices, or an intrinsic property of the device itself.

We adopt the modeling approach described in \cite{guImpedanceBasedWholeSystemModeling2021}. The small-signal admittance model of a generic converter is depicted in Figure~\ref{fig:frame_embedding}. The input to the model is the voltage at the point of common coupling, $V_{DQ}$, expressed in a local steady $DQ$ reference frame, while the output is the converter current $I_{DQ}$ in the same frame. The converter controller operates in a local swing $dq$ reference frame, and is conceptually divided into two parts: 
\begin{itemize}
    \item Synchronization controller $K_v(s)$, which generates the angle $\epsilon$ used to transform variables from the local steady to the swing frame.
    \item Remainder of the control system $Y_{dq}(s)$, which includes the current control, voltage control, virtual impedance, reactive power control, and other controllers not related to the synchronization.
\end{itemize}
 The closed-loop admittance in the local steady frame can then be written as~\eqref{eq:frame_embedding}. The vectors $V_0^e=[-V_{q,0},V_{d,0}]^T$ and $I_0^e=[-I_{q,0},I_{d,0}]^T$ represent the linearized rotation terms between the two frames, where $V_{d,0}$, $V_{q,0}$, and $I_{d,0}$, $I_{q,0}$ are the voltage and current operating points.
\begin{equation}
    Y_{DQ}(s)=(Y_{dq}(s)+I_0^eK_v(s))(I+V_0^eK_v(s))^{{-}1}
    \label{eq:frame_embedding}
\end{equation}
 \begin{figure}
     \centering
     \includegraphics[width=\linewidth]{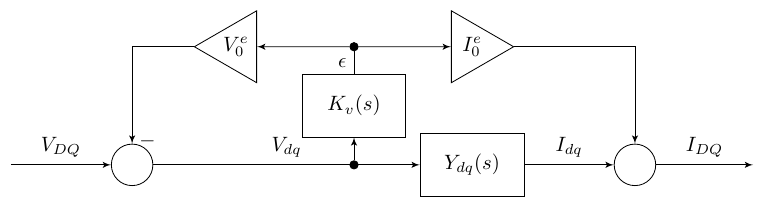}
     \caption{Small-signal model of a generic converter with synchronization frame embedding.}
     \label{fig:frame_embedding}
 \end{figure}
For a typical GFM converter with power synchronization, the angle $\epsilon$ is generated by a dynamic system of the form~\eqref{eq:freq_tf}, where $\omega$ denotes the frequency, $P$ is the active power, and $H_P(s)$ is the transfer function of the power-based synchronization loop. 
\begin{equation}
    \omega = H_P(s)P\,,\quad \epsilon=\frac{1}{s}\omega
    \label{eq:freq_tf}
\end{equation}
Linearization of the power expression yields \eqref{eq:power_lin}, in which $P$ depends only on the voltage in the local steady frame, allowing $K_v(s)$ to be written as~\eqref{eq:Kv}.
\begin{equation}
    P \approx \underbrace{(V_0 Y_{dq}+I_o)}_{:=G_P(s)} V_{dq}\,,\, V_0=[V_{d,0} \, V_{q,0}], \, I_0=[I_{d,0} \, I_{q,0}]
    \label{eq:power_lin}
\end{equation}
\begin{equation}
    K_v(s) = \frac{H_P(s)}{s}G_P(s)
    \label{eq:Kv}
\end{equation}
Since our focus is on the low-frequency behavior, we consider the DC gain of the admittance at $s=j0$. The following proposition summarizes the key result.
\begin{prop}
    Assume $V_{q,0}=0$, $V_{d,0}\neq0$, and $K_v(s)V_0^e$ is not identically zero. The converter admittance matrix \eqref{eq:frame_embedding} with $K_v(s)$ as in \eqref{eq:Kv} evaluated at $s=0$ is equal to:
\begin{equation}
    Y_{DQ}(0)=\begin{bmatrix}
        -\frac{I_{d,0}}{V_{d,0}} & -\frac{I_{q,0}}{V_{d,0}} \\
        \beta & \frac{I_{d,0}}{V_{d,0}}
    \end{bmatrix}
    \label{eq:Y_DQ_0}
\end{equation}
for some $\beta$. Moreover, $Y_{DQ}(0)$ is not sectorial.
\label{prop:Y_DQ_0}
\end{prop}
\begin{proof}
    The derivation of \eqref{eq:Y_DQ_0} is given in Appendix \ref{apx:Prop_Y_DQ_0}. 
    Since the trace is null, then the eigenvalues' sum is zero. The numerical range of a $2\times2$ matrix is an ellipse with the two eigenvalues as foci. Depending on the values of $\beta$ and $I_{q,0}$, the ellipse has as principal axis either the real axis or the imaginary axis. In both cases, the origin belongs to the numerical range, making the matrix not sectorial.
\end{proof}

The assumption $V_{q,0}=0$ is without loss of generality since a rotation matrix can always be chosen so that this holds. Furthermore, depending on the converter operating point and the parameters, not only the admittance is non-sectorial but also non-quasi-sectorial.
This result demonstrates that, regardless of the control implementation, sectoriality cannot be achieved at 0 Hz and will generally not be satisfied at low frequencies. The root cause lies in the fact that a GFM converter behaves as a constant power load/source. Even when frequency droop is implemented (so that power varies with frequency), from the admittance perspective (voltage-to-current relation) the converter still behaves as a constant power device.
Furthermore, the above analysis shows that the converter is not passive\footnote{Here, “passive” is used in the loose frequency-wise sense common in the power electronics literature, referring to positive realness of $Y(j\omega)$.} at low frequency. This finding is consistent with previous work \cite{deyAnalysisPassivityCharacteristics2021}, but is more general in scope as it does not depend on a specific control implementation or parameter values.

\begin{rem}
    A similar analysis can be carried out for grid-following (GFL) converters. For instance, with an SRF-PLL synchronization loop, 
    $K_v(s)$ takes the form $\frac{H_{PLL}(s)}{s}[0,1]$.
    %$G_P(s)$ is replaced with the vector $[0,1]$. 
    For a typical PI-based active power control, the resulting admittance has the same form as in \eqref{eq:Y_DQ_0}. 
    Therefore, the same considerations regarding low-frequency non-sectoriality and lack of passivity apply to GFL converters as well.
\end{rem}
\section{Proposed Approach}
\label{sec:propo_approch}
Following the rationale in \cite{deyPassivityBasedDecentralizedCriteria2023} and \cite{chenExtendedFrequencyDomainPassivity2025a}, we show that sectoriality is not an intrinsic property of the system, but rather depends on the specific selection of inputs and outputs. Consequently, an alternative choice of variables may enable the system to exhibit sectoriality.

\subsection{Power-Polar coordinates frame}
A particularly suitable choice for GFM converters is to represent voltages in power-polar coordinates, i.e., using voltage magnitude and phase as inputs, and active and reactive power as outputs. As demonstrated in \cite{deyAnalysisPassivityCharacteristics2021}, this representation yields a passive mapping for GFM converters in the low-frequency range, and therefore a sectorial system. For each converter $i$, the associated linearized change of variables is given by \eqref{eq:change_of_coord_1}–\eqref{eq:change_of_coord_2}, where the transformation matrices $\mathcal{E}_i$, $\mathcal{F}_i$, and $\mathcal{C}_i$ depend on the voltage and current operating point.
\begin{equation}
    \begin{bmatrix}
        \phi_i \\ |V|_i 
    \end{bmatrix} =
    \underbrace{\begin{bmatrix}
        -V^i_{q,0} & \frac{V^i_{d,0}}{\sqrt{(V^i_{d,0})^2+(V^i_{q,0})^2}} \\
        V^i_{q,0} & \frac{V^i_{q,0}}{\sqrt{(V^i_{d,0})^2+(V^i_{q,0})^2}}
    \end{bmatrix}}_{:=\mathcal{F}_i^{-1}}
    \begin{bmatrix}
        V_{D,i} \\ V_{Q,i} 
    \end{bmatrix}
    \label{eq:change_of_coord_1}
\end{equation}

\begin{equation}
    \begin{bmatrix}
        P_i \\ Q_i
    \end{bmatrix} = \underbrace{\begin{bmatrix}
        V^i_{d,0} & V^i_{q,0}  \\
        V^i_{q,0} & -V^i_{d,0}
    \end{bmatrix}}_{:=\mathcal{E}_i} \begin{bmatrix}
        I_{D,i} \\ I_{Q,i} 
    \end{bmatrix} + 
    \underbrace{\begin{bmatrix}
        I^i_{d,0} & V^i_{q,0}  \\
        -I^i_{q,0} & I^i_{d,0}
    \end{bmatrix}}_{:=\mathcal{C}_i}
    \begin{bmatrix}
        V_{D,i} \\ V_{Q,i} 
    \end{bmatrix}
    \label{eq:change_of_coord_2}
\end{equation}

With these definitions, the converter’s new input–output relation becomes as in \eqref{eq:pp_trasnf} and similar the transformation is applied to the network, yielding \eqref{eq:pp_trasnf_net}.
\begin{equation}
    J_{C,i}(s) :=(\mathcal{E}_iY_{C_i}(s)+\mathcal{C}_i)\mathcal{F}_i
    \label{eq:pp_trasnf}
\end{equation}
\begin{equation}
    \mathbf{J}_{net}(s) :=(\boldsymbol{\mathcal{E}}\mathbf{Y}_{net}(s)-\boldsymbol{\mathcal{C}})\boldsymbol{\mathcal{F}}
    \label{eq:pp_trasnf_net}
\end{equation}
$\boldsymbol{\mathcal{E}}$, $\boldsymbol{\mathcal{C}}$ and $\boldsymbol{\mathcal{F}}$ are the block-diagonal matrices build from $\mathcal{E}_i$, $\mathcal{C}_i$ and $\mathcal{F}_i$, respectively. The minus sign in $\mathcal{C}$ comes from the fact that the current in the network flows in the opposite direction to that assumed in the converter. We denote $\mathbf{J}_C$ the block-diagonal matrix build from $J_{C_i}$ as in \eqref{eq:block_diagonal}.

This change of variables can be interpreted as a loop transformation, as illustrated in Fig. \ref{fig:loop_shaping}. 
Assuming the transformation is well-defined, input/output stability of the transformed system and the original system are equivalent. Indeed, the stability is determined by the zeros of $\det(\mathbf{Y}_{net}(s)+\mathbf{Y}_C(s))$. Since $\det(\mathbf{J}_C(s) +\mathbf{J}_{net}) = \det(\boldsymbol{\mathcal{E}}(\mathbf{Y}_{net}(s)+\mathbf{Y}_{C}(s))\boldsymbol{\mathcal{F}})$, the zeros of the transformed system and the original are identical. Moreover, if there are no right plane pole/zero cancellations, then
 input–output stability implies internal stability, as shown for instance in \cite{häberle2025decentralizedparametricstabilitycertificates}. Theorem \ref{thm:dec_mixed_small_gain_phase} can be applied directly to the transformed system without loss of generality.

\begin{figure}
    \centering
    \includegraphics[width=\linewidth]{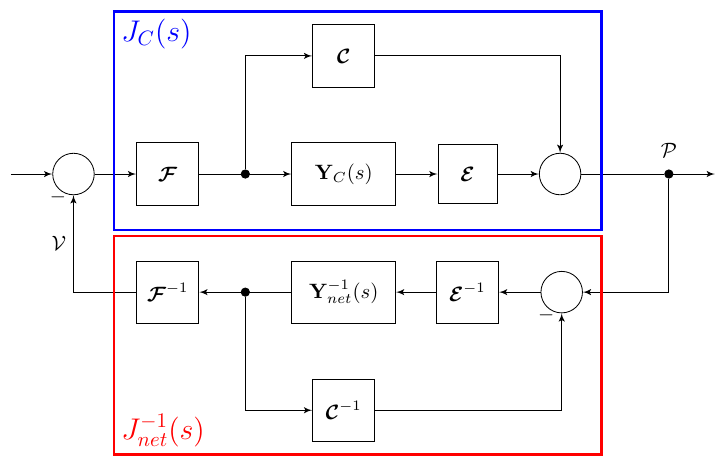}
    \caption{Proposed transformation as loop-shaping transformation}
    \label{fig:loop_shaping}
\end{figure}

\begin{rem}
    The proposed transformation 
    %generalizes the post-multiplication approach introduced in 
    combines the post-multiplication and the shaping functions from \cite{woolcockMixedGainPhase2023, chenExtendedFrequencyDomainPassivity2025a, deyPassivityBasedDecentralizedCriteria2023}. In \eqref{eq:pp_trasnf}, both pre- and post-multiplication are performed, along with a translation term $\mathcal{C}$.
\end{rem}
Taking in consideration the GFM model introduced in \eqref{eq:frame_embedding} and applying the transformation, a straightforward computation shows that $J_C(s)$ evaluated at 0 is equal to \eqref{eq:J_C_0}, for some scalar $\gamma$. Indeed, in this formulation a phase jump does not produce any steady-state power deviation, while a deviation in voltage magnitude results only in a change in reactive power. 
It is a diagonal matrix with eigenvalues $0$ and $\gamma$, and its numerical range is a straight segment connecting this two eigenvalues. 
\begin{equation}
    J_C(0)=\begin{bmatrix}
        0 & 0 \\
        0 & \gamma \\
    \end{bmatrix}
    \label{eq:J_C_0}
\end{equation}

In this frame, the converter admittance at zero frequency is quasi-sectorial with a well-defined phase. While this may not extend across the low-frequency range, it remains a promising indication.
Nevertheless, for the small-phase criterion to be applicable, sectoriality must hold across the entire frequency range for both $\mathbf{J}_C(s)$, and $\mathbf{J}^{-1}_{net}(s)$. 
While the transformation can resolve the low-frequency non-sectoriality, it does not guarantee that sectoriality will persist at higher frequencies.

To illustrate this point, the same analysis as before is repeated for a GFM connected to an infinite bus. The small-phase criteria is shown in Figure \ref{fig:pp_frame_cond}. The converter exhibits sectorial behavior up to approximately $64$ Hz, within which the small-phase condition is satisfied. Beyond this frequency, neither the converter nor the network remain sectorial, rendering the criterion inapplicable. This highlights that, while the transformation addresses the dominant low-frequency limitation, it does not fully eliminate the applicability restrictions of the small-phase approach.

\begin{rem}
    Although $\mathbf{J}_C(s)$ inherits stability from $\mathbf{Y}_C(s)$, the same cannot be guaranteed for $\mathbf{J}^{-1}_{net}(s)$. The associated transformation can be interpreted as a feedback operation, which may potentially destabilize the system. Consequently, prior to applying Theorem \ref{thm:dec_mixed_small_gain_phase}, the stability of $\mathbf{J}^{-1}_{net}(s)$ must be verified.
\end{rem}

\begin{figure}
    \centering
    \includegraphics[width=\linewidth]{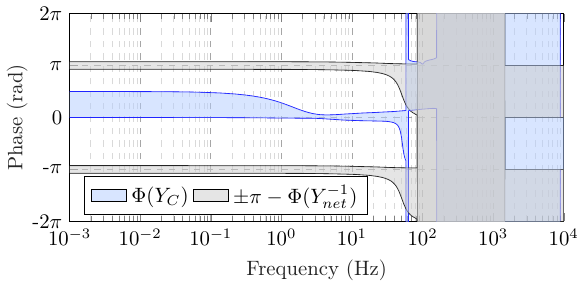}
    \caption{Small phase condition for a GFM converter in Power-Polar frame}
    \label{fig:pp_frame_cond}
\end{figure}

\subsection{Unified Approach}
As shown in the previous section, neither the criteria in the standard rectangular admittance frame nor those in the power–polar frame are satisfied across the entire frequency range.
In the rectangular frame, the phase remains well-defined and the small-phase condition is fulfilled at high frequencies (e.g., $f>0.7$ Hz). In contrast, for the loop-shaped system in the polar frame, the phase is well-defined and the small-phase condition is satisfied at low frequencies (e.g., $f<64$ Hz).

An intuitive approach might be to combine the two frames: use the rectangular frame phase at high frequencies and the polar-frame phase at low frequencies, and declare the system stable if the small-phase criterion is satisfied in at least one frame at each frequency. However, it is possible to construct a counterexample demonstrating that stability cannot be inferred simply by merging frequency-wise results from different frames.
The fundamental reason lies in the fact that different transformations bound different quantities:
the eigenvalues of $\mathbf{J}_C(j\omega)\mathbf{J}^{-1}_{net}(j\omega)$ are, in general, not the same as those of $\mathbf{Y}_C(j\omega)\mathbf{Y}^{-1}_{net}(j\omega)$ if $\mathcal{C}$ is different from the zero matrix.

%One possible solution is to define a generic transformation, parameterized by the matrices $\mathcal{E}_i$, $\mathcal{F}_i$, and $\mathcal{C}_i$ such that sectoriality is ensured across the entire frequency range. This concept is further extended to frequency-dependent matrices, as expressed in \eqref{eq:generic_trasnf}.

To merge these two frames into a single system, we consider the shaping matrices $\mathcal{E}_i$, $\mathcal{F}_i$, and $\mathcal{C}_i$ to be frequency-dependent matrices, as expressed in \eqref{eq:generic_trasnf}.

\begin{equation}
\begin{aligned}
    &\tilde{J}_{C_i}(s) =(\mathcal{E}_i(s)Y_{C_i}(s)+\mathcal{C}_i(s))\mathcal{F}_i(s) \\
    &\tilde{\mathbf{J}}_{net}(s) =(\boldsymbol{\mathcal{E}}(s)\mathbf{Y}_{net}(s)-\boldsymbol{\mathcal{C}}(s))\boldsymbol{\mathcal{F}}(s)
    \label{eq:generic_trasnf}
\end{aligned}
\end{equation}
These matrices must be selected so that each $\tilde{J}_{C_i}(s)$ and $\tilde{\mathbf{J}}^{-1}_{net}(s)$ are stable.
Moreover, if the shaping matrices $\boldsymbol{\mathcal{E}}(s)$, $\boldsymbol{\mathcal{C}}(s)$ and $\boldsymbol{\mathcal{F}}(s)$ are well-posed, the stability of the loop-shaped system composed of $\tilde{\mathbf{J}}_{c}$ and $\tilde{\mathbf{J}}_{net}$ is equal to the stability of the original system composed by $\mathbf{Y}_{c}$ and $\mathbf{Y}_{net}$. 
This formulation provides significant flexibility in choosing the transformation matrices, but identifying ones that simultaneously ensure open-loop stability and sectoriality over the full frequency range is nontrivial.

%Therefore, we propose a different formulation with an added degree of freedom that helps us satisfy the sectoriality assumption for practical examples"  + getting good bounds 

Ideally, at low frequencies the matrices should correspond to those in \eqref{eq:change_of_coord_1}-\eqref{eq:change_of_coord_2} (denoted from now on as $\mathcal{E}_{J,i}$, $\mathcal{F}_{J,i}$ and $\mathcal{C}_{J,i}$), while at high frequencies they should revert to the rectangular frame (i.e., $\mathcal{E}=\mathcal{F}=I$ and $\mathcal{C}=0$, denoted from now on as $\mathcal{E}_{Y,i}$, $\mathcal{F}_{Y,i}$ and $\mathcal{C}_{Y,i}$). A continuous transition between these two regimes is necessary. A trivial choice for $\mathcal{E}_i(s)$, and similar for $\mathcal{F}_i(s)$ and $\mathcal{C}_i(s)$, can be as in \eqref{eq:freq_dep_tf} using an low/high-pass filtering approach. 
The filters transfer function are given as: $H_{LPF}(s)=w_c/(s+w_c)$ and $H_{HPF}(s)=1-H_{LPF}(s)$, with $w_c$ the filter cut-off frequency. 
\begin{equation}
    \mathcal{E}_i(s) = H_{LPF}(s)\mathcal{E}_{J,i}+H_{HPF}(s)\mathcal{E}_{Y,i}
    \label{eq:freq_dep_tf}
\end{equation}
This transformation proves in general inadequate since sectoriality is lost near the filter cut-off frequency $w_c$. 
%This is because when the filter start to effect the transformed system, the numerical range starts to rotate and at some point it incorporates the origin. 
To ensure sectoriality while obtaining tight phase bounds, we propose instead the transformation in~\eqref{eq:freq_dep_tf_better}, with an additional degree of freedom $W_i$. In this formulation, frequency-dependent filtering is applied exclusively to $\mathcal{C}$ and $\mathcal{F}(s)$, while $\mathcal{E}(s)$ remains constant.
\begin{equation}
\begin{aligned}
    &\mathcal{E}_i(s) = \mathcal{E}_{J,i} \\
    &\mathcal{C}_i(s) = H_{LPF}(s)\mathcal{C}_{J,i}\\
    &\mathcal{F}_i(s) = H_{LPF}(s)\mathcal{F}_{J,i}+H_{HPF}(s)W_i(s)\mathcal{E}_{J,i}^{-1}
\end{aligned}
\label{eq:freq_dep_tf_better}
\end{equation}
%Assume for the moment that $W_i(s)$ is the identity matrix. The rationale behind the particular choice, even if we do not filter $\mathcal{E}_i$, at high frequency we obtain $J_{C_i}\approx\mathcal{E}_{J,i}Y_{C_i}\mathcal{E}_{J,i}^{-1}$, and since a similarly transformation does not change the numerical range, we get back the rectangular frame.
From practical observation, this transformation compared to \eqref{eq:freq_dep_tf} provides smooth transitions between the two frames, and sectoriality is easier to obtain.
For instance, setting $W_i(s)$ equal to the identity matrix
let us recover the rectangular frame at high frequencies, since $J_{C_i}\approx\mathcal{E}_{J,i}Y_{C_i}\mathcal{E}_{J,i}^{-1}$ (and a similarly transformation does not change the numerical range). $W_i(s)$ is an additional weighted matrix  introducing further flexibility to, for example, transfer known dynamics from $Y_{C_i}(s)$ to $\mathbf{Y}_{net}(s)$, eventually allowing better sectorial behavior and stricter phase bounds. This concept is clarified in the next section with an example. Nonetheless, other transformations may ensure sectoriality and tight phase bounds; finding out an optimal loop-shaping transformation is still an open point.
\section{Example 1 - Infinite Bus}
\label{sec:ex_1}
In this section, we assess how the proposed methods can be applied to study the stability of a converter connected to the infinite bus system described earlier. Moreover, we address the conservativeness limitation showed in Section \ref{sec:lim} and how this can be resolved with a proper choice of the weighting matrix $W_i(s)$.

\subsection{Virtual Admittance compensation}
Figure \ref{fig:GFM_gain_phase_test} shows that, even when both systems are sectorial, the small-phase eigenvalue bounds are highly conservative for frequencies below $50$ Hz. This conservatism arises from the fact that both of $Y_C$ and $Y^{-1}_{net}$ exhibit large phases, and, recalling \eqref{eq:phase_bound}, the small-phase bound is derived under a worst-case assumption, pairing maxima with maxima and minima with minima. However, favorable phase interactions may occur in the product $Y_CY^{-1}_{net}$, leading to partial compensation and a reduction of its eigenvalues' phase. This phenomenon is evident in Figure~\ref{fig:phase-bounds_Yv}, where the small-phase bounds $\Phi(Y_C)+\Phi(Y_{net}^{-1})$ clearly overestimate the actual eigenvalues' phase $\angle\lambda(Y_CY_{net}^{-1})$.
\begin{figure}
    \centering
    \includegraphics[width=\linewidth]{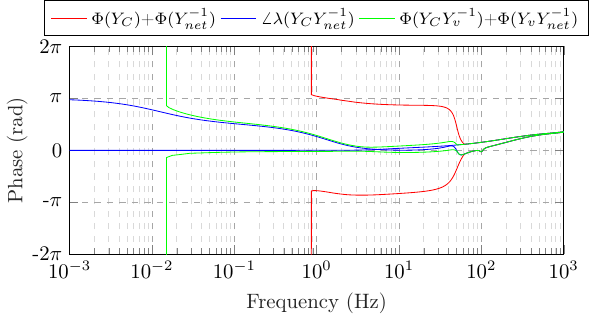}
    \caption{Small-phase bounds without and with virtual admittance compensation.}
    \label{fig:phase-bounds_Yv}
\end{figure}

A tighter bound can be obtained by reducing the large phase spread of the numerical range, and therefore limiting the extreme combinations of phases. Since the converter implements a virtual admittance control $Y_v(s)$, and its behavior around $50$ Hz is mainly determined by the virtual admittance, we can post-multiply $Y_C(s)$ by the inverse of $Y_{v}(s)$, such that $Y_{C}(s)Y_{v}(s)^{-1}\approx I$ around those frequencies. The network admittance is transformed accordingly to $Y_{net}(s)Y_{v}(s)^{-1}$.
%we select the post-multiplication matrix $\mathcal{F}(s)=Y_{virt}(s)^{-1}$ (with $\mathcal{E}=I$, $\mathcal{C}=0$). So that we obtain $\tilde{J}_c(s)\approx I$. Meanwhile, $\tilde{J}_{net}(s)= Y_{net}(s)Y_{virt}(s)^{-1}$. 
From another viewpoint, a known dynamic behaviour is shifted from the converter to the network, thereby avoiding the conservatism. In~\eqref{eq:freq_dep_tf_better}, this correspond in the the selection $W_i(s)=Y_v(s)^{-1}$.

Figure~\ref{fig:phase-bounds_Yv} illustrates the effect of this compensation. It can be seen that applying the proposed virtual admittance filtering leads to noticeably improved bounds.
%This figure also highlights the necessity of the translation term $\mathcal{C}$. Since one eigenvalue’s phase approaches $\pi$ at low frequencies and pre-/post-multiplication does not alter it, we must rely on the small gain at those frequencies to ensure stability. Conversely, with a properly chosen $\mathcal{C}$, the stability assessment can rely solely on the phase criterion, as discussed in the next subsection.

At first glance, this approach may appear to undermine the decentralized nature of the certificate. However, in practice, exact knowledge of the virtual admittance values (that is, the behavior around $50$Hz) is not required, and approximate values are sufficient, as shown afterwards. In a decentralized setting, these values can either be selected apriori as fixed parameters or assigned differently across units by the stability certifier. In any case, proposed grid codes for grid-forming converters already define fixed ranges for the effective impedance~\cite{entsoe}. 
\begin{rem}
%Notice in Figure \ref{fig:phase-bounds_Yv}, how one eigenvalues tends to $\pi$ for frequency going to zero, and at the same time small-gain does not guarantee it has module less than one, than stability is difficult to assest. Therefore, here it is evident the need of shift introduced by the $\mathcal{C}$, which can place the eigenvalues in better position where bounding it is easier.
%\color{red}
Figure~\ref{fig:phase-bounds_Yv} shows one eigenvalue approaches $\pi$ as the frequency tends to zero. In this region, the small-gain condition does not ensure $|\lambda|<1$, rendering stability assessment inconclusive. This motivates the shift introduced by $\mathcal{C}$, which relocates the eigenvalues of the open-loop transfer function in a more suitable region where small-phase allows to certify stability.% \color{black}
\end{rem}
%In Figure~\ref{fig:phase-vi-filt}, the small-phase criteria is shown when the post-filter $Y_{virt}(s)$ is used. We can notice that, compared to Figure \ref{fig:GFM_gain_phase_test}, the phase spread of each system is reduced and the phase margin between converter and network is increased.

%\begin{figure}[tbh]
%    \centering
%    \includegraphics[width=\linewidth]{Figures/GFM_InfBus_Yvirt/Phase.pdf}
%    \caption{Small phase condition of a GFM converter with virtual admittance post-filtering.}
%    \label{fig:phase-vi-filt}
%\end{figure}
\subsection{Full transformation}
In Figure~\ref{fig:phase-dep-frame}, the small-phase criteria is shown with the frequency-depended transformation \eqref{eq:generic_trasnf}, with matrices defined as in~\eqref{eq:freq_dep_tf_better}. The cut-off frequency of the filters is set to $0.5$ Hz.  Compared to the previous results, we can obtain sectoriality for the full frequency range and the closed-loop system stability is guaranteed by the small-phase conditions alone.
%\color{red}
This means that the stability of this GFM converter does not depend on the grid strength, meant as its gain, and stability is preserved under both weak and strong conditions. In this sense, the small-phase condition provides a grid-strength-independent stability guarantee. However, variations in the network characteristics, such as changes in the $R/X$ ratio or loading, may alter the network phase. Such variations can reduce the phase margin and potentially drive the converter toward instability.
%\color{black}

\begin{figure}
    \centering \includegraphics[width=\linewidth]{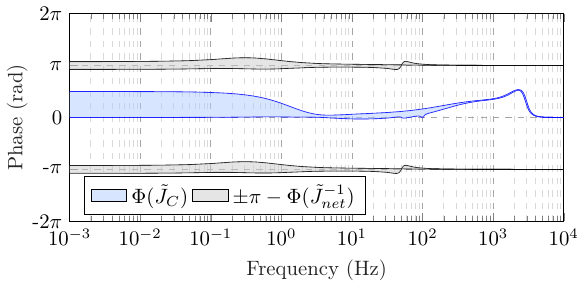}
    \caption{Small phase condition of a GFM converter with frequency-dependent transformation \eqref{eq:generic_trasnf}, \eqref{eq:freq_dep_tf_better}.}
    \label{fig:phase-dep-frame}
\end{figure}

\subsection{Unstable case}
To further validate the method, we examine the sufficient condition of the theorem: if the closed-loop system is unstable, then the small-phase and the small-gain condition must be simultaneously violated at some frequency. In practice, destabilizing a grid-forming converter is not straightforward. Following \cite{de2012automatic}, we introduce a purely resistive load at the PCC and increase the infinite bus impedance. By switching the control mode to constant-Q and reducing the VSM damping of the GFM from $0.5$ to $0.05$, the closed-loop becomes unstable, exhibiting sustained oscillations at 1.2 Hz. Figure \ref{fig:GFM_Unstable} illustrates the mixed gain/phase test, considering the proposed transformation. 
Enabling the reactive power control reduces the gain at low frequencies, leading to the small-gain condition to be satisfied.
However, while in the stable case ($D=0.5$) the small-phase condition is satisfied, in the unstable case ($D=0.05$), the converter phase increases substantially, and together with the network phase causes the phase condition to be violated, while the gain condition is also not satisfied around the unstable oscillation frequency.
\begin{figure}
    \centering
    \begin{subfigure}{\columnwidth}
        \centering
        \includegraphics[width=\linewidth]{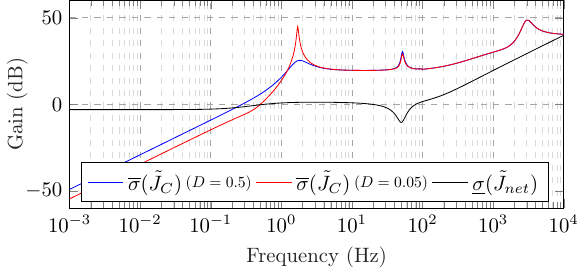}
    \end{subfigure}
    \begin{subfigure}{\columnwidth}
        \centering
        \includegraphics[width=\linewidth]{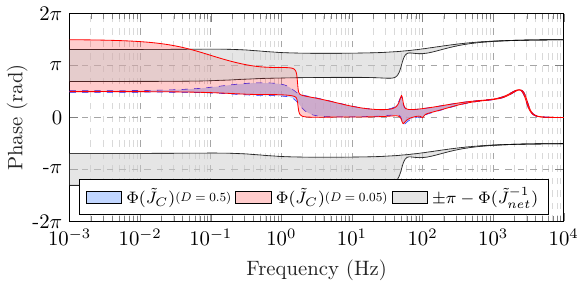}
    \end{subfigure}
    \caption{GFM Infinite Bus - Stable ($D=0.5$) vs Unstable ($D=0.05$).}
    \label{fig:GFM_Unstable}
\end{figure}

\section{Example 2 - IEEE 14 Bus System}
\label{sec:ex_2}
\begin{figure}
    \centering
    \includegraphics[width=\linewidth]{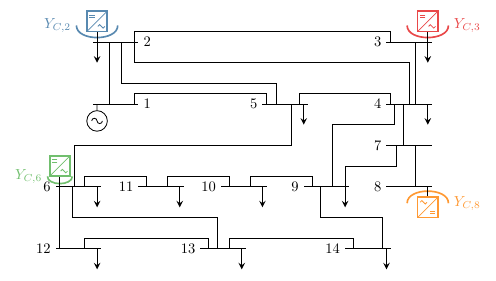}
    \caption{IEEE 14 Bus System diagram}
    \label{fig:ieee14_diagram}
\end{figure}
To extend and show the scalability of our approach, we consider the modified IEEE 14 Bus System shown in Figure \ref{fig:ieee14_diagram}, which parameters are given in \cite{RepoZenodo}. Four GFM converters are connected at the buses 2, 3, 6, and 8.
Similarly to the infinite bus example, to build an unstable case, we select a low damping value for converter 8, connect purely resistive loads, and increase the line 7-8 impedance. In this setup, the system is unstable with an unstable oscillation at 0.54 Hz driven by converter 8.
In Figure \ref{fig:IEEE14_criteria}, the decentralized conditions are shown using the proposed transformation and virtual admittance compensation. While the converters 2, 3, and 6 satisfy the small-phase condition, converter 8 does not. In particular, below 1.64 Hz, its phase overlaps the network phase. This does not imply the instability of the system; however, knowing that the system is unstable, it identifies the responsible converter, and by appropriately retuning its control parameters, stability can potentially be restored. For example, if we apply the same tuning of converter 6, which satisfy the small-phase condition, to converter 8, then stability is recovered.

In this example, each converter implements a different $R/X$ ratio in its virtual admittance, as reported in Table~\ref{tab:virt_adm}, while a fixed value ($R_v=0.01$, $X_v=0.1$) is used for the virtual admittance compensation for all converters. This demonstrates that, even when the exact virtual admittance is unknown, a reasonable estimate can still be used to obtain adequate phase bounds and, consequently, an effective stability validation with the proposed approach.
\begin{figure}
    \centering
    \begin{subfigure}{\columnwidth}
        \centering
        \includegraphics[width=\linewidth]{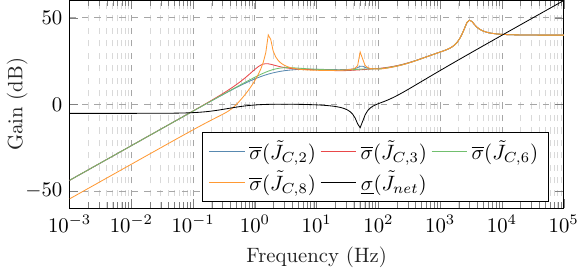}
    \end{subfigure}
    %\vspace{-1em} 
    \begin{subfigure}{\columnwidth}
        \centering
        \includegraphics[width=\linewidth]{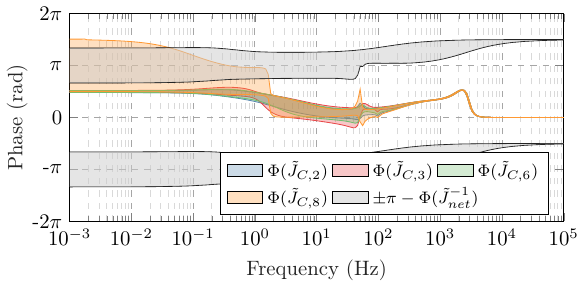}
    \end{subfigure}
    \caption{Stability analysis of the 14 Bus System}
    \label{fig:IEEE14_criteria}
\end{figure}

\begin{table}
\centering
\caption{Virtual admittance setting in 14 Bus System}
\label{tab:virt_adm}
\begin{tabular}{ccc}
GFM Bus & $R_v$ (p.u) & $X_v$ (p.u) \\
\hline
$2$ & $0.02$ & $0.1$ \\
$3$ & $0.05$ & $0.1$ \\
$6$ & $0.03$ & $0.1$ \\
$8$ & $0.01$ & $0.1$ \\
\hline
\end{tabular}
\end{table}
\section{Conclusion}
\label{sec:conc}
%This paper has extended decentralized stability conditions for power systems by introducing a transformation that preserves decentralization while addressing the lack of sectoriality at low frequencies in standard converters. By switching from the rectangular to the polar admittance frame, improved low-frequency phase behavior is achieved, and the combination of both frames provides a balanced solution across frequencies.

This work revisited decentralized small gain–phase conditions for power systems and identified their inherent limitations when applied to grid-forming converters. To overcome this, a generalized loop-shaping transformation was introduced, combining the rectangular and the polar admittance frames. Not only this addresses the lack of sectoriality at low frequencies in converters, but also renders tighter bounds for the phase behavior across the whole frequency range.

The method was demonstrated on a grid-forming converter connected to an infinite bus, where virtual admittance filtering further reduced conservativeness, and on the IEEE 14-bus system, where it successfully identified unstable devices. 

While the results are promising, challenges remain, particularly when studying systems that include not only grid-forming converters but also other devices, such as synchronous machines and grid-following converters. Synchronous machines have similar dynamics to the GFM converter considered here and we expect the proposed approach to work. However, the impact of auxiliary control loops (e.g. exciters, automatic voltage regulators, power system stabilizers) requires further verification. 
For grid-following converters, which exhibit distinct dynamics, the small-gain condition plays a more critical role, and the benefits of the proposed transformation are limited. 
To obtain a less conservative gain condition for GFL units, it would be beneficial to shift the high gain of the GFMs toward the network, thereby enhancing its strength.

%This makes it necessary to shift the high gain of GFMs toward the network to enhance its strength, thereby obtaining less conservative gain conditions for GFL units.

%While the results are promising, challenges remain, particularly when studying systems that include both grid-forming and grid-following converters, as well as in guaranteeing stability and sectoriality of the grid impedance. 

%We intentionally excluded GFL converters. Although the concepts presented are general and applicable to different control structures, GFL converters exhibit distinct dynamics that require a separate, more detailed analysis and potentially a different definition of the transformation matrices. Furthermore, as highlighted in [3], the small-gain condition plays a more critical role, making it necessary to shift the high gain of GFMs toward the network to enhance its strength.

%While the results are promising, challenges remain, particularly when studying systems that include both grid-forming and grid-following converters, as well as in guaranteeing stability and sectoriality of the grid impedance. 

Future work will focus on systematic selection of transformation matrices and on reducing conservativeness by embedding network topology information into the method.

\bibliographystyle{IEEEtran}
\bibliography{DecCondition_bib}

\appendices
\section{GFM Control structure}
In this paper, we consider the GFM control structure represented in Figure \ref{fig:gfm_diag}. The controller are implemented in the stationary reference frame $\alpha\beta$. It consists of a Proportional-Resonant $PR(s)$ current controller, given in~\eqref{eq:pr_tf}, and a Virtual Admittance $Y_v(s)$ voltage controller, given in~\eqref{eq:yv_tf}, with virtual resistance $R_v$ and virtual inductance $L_v$.
\begin{equation}
    PR(s) = K_p+\frac{K_i}{s^2+(2\pi50)^2}
    \label{eq:pr_tf}
\end{equation}
\begin{equation}
    Y_v(s)=\frac{1}{sL_v+R_v}, \quad L_v=\frac{X_v}{2\pi50}
    \label{eq:yv_tf}
\end{equation}
The synchronization control is implemented through a Virtual Synchronous Machine (VSM) with inertia $H$ and damping $D$. A Proportional-Integral (PI) controller regulates the reactive power changing the reference voltage magnitude. Unless explicitly stated, the reactive power controller is disabled, and a constant voltage magnitude reference is used. Parameters and further details are available in \cite{RepoZenodo}.

\begin{figure}
    \centering
    \includegraphics[width=0.9\columnwidth]{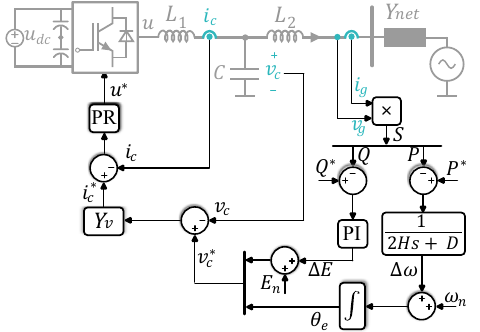}
    \caption{GFM control Structure}
    \label{fig:gfm_diag}
    \vspace{-1em}
\end{figure}

\label{apx:GFM_control}

\section{Proof of Proposition \ref{prop:Y_DQ_0}}
\label{apx:Prop_Y_DQ_0}
The expression of $Y_{DQ}(s)$ in~\eqref{eq:frame_embedding} can be rewritten as in \eqref{eq:A1}, and further simplify as in \eqref{eq:A2}.
\begin{align}
    Y_{DQ}(s)&=(Y_{dq}(s)+I_0^eK_v(s))\left(I-\frac{V_0^eK_v(s)}{1+K_v(s)V_0^e} \right)\label{eq:A1}\\
    &=Y_{dq}(s)-(Y_{dq}(s)V_0^e-I_0^e)\frac{K_v(s)}{1+K_v(s)V_0^e}\label{eq:A2}
\end{align}
Considering $K_v(s)=\frac{H_P(s)}{s}G_P(s)$ and taking the limit for $s$ tends to zero we obtain~\eqref{eq:A3}.
\begin{equation}
    Y_{DQ}(0)= Y_{dq}(0)-(Y_{dq}(0)V_0^e-I_0^e)\frac{G_P(0)}{G_P(0)V_0^e}
    \label{eq:A3}
\end{equation}
By post-multiplying by $V_0^e$, we obtain $Y_{DQ}(0)V_0^e=-I_0^e$, i.e. the second column of $Y_{DQ}(0)$ times $V_{d,0}$ is equal to $-I_0^e$, as given in~\eqref{eq:Y_DQ_0}.
Through straightforward but lengthy algebraic manipulations (omitted here for brevity), it can be shown that the (1,1) entry of $Y_{DQ}(0)$ also matches the corresponding expression in \eqref{eq:Y_DQ_0}.

\begin{comment}
\end{comment}

\end{document}